\documentclass[11pt]{article}
\usepackage[utf8]{inputenc}
\usepackage[a4paper,left = 1in,right = 1in, top = 1.5in, bottom = 1.2in]{geometry}
\usepackage{amsmath}
\usepackage{amsthm}
\usepackage{palatino}
\usepackage{bm}
\usepackage{dsfont}
\usepackage{cancel}
\usepackage{amssymb}
\usepackage{commath}
\usepackage{mathtools}
\usepackage{float}
\usepackage{accents}
\usepackage{algorithm}
\counterwithin{equation}{section}
\counterwithin{algorithm}{section}

\usepackage[numbers]{natbib}
\usepackage{subcaption}
\usepackage{graphicx}
\usepackage{color}
\usepackage{tcolorbox}
\newtcolorbox{mybox}[3][]
{
colframe = black,
  colback  = black,
  coltitle = white,  
  coltext = white,
  title    = {#3},
  #1,
}

\usepackage{xcolor}
\usepackage[colorlinks,linkcolor=blue,citecolor=blue,urlcolor = violet]{hyperref}
\usepackage[capitalise]{cleveref}

\def\be{\begin{equation}}
\def\ee{\end{equation}}

\def\bea{\begin{align}}
\def\eea{\end{align}}

\def\bea*{\begin{align*}}
\def\eea*{\end{align*}}
\theoremstyle{plain}
\newtheorem{theorem}{Theorem}[section]

\newtheorem{lemma}[theorem]{Lemma}

\newtheorem{proposition}[theorem]{Proposition}

\newtheorem{definition}[theorem]{Definition}
\newtheorem{example}[theorem]{Example}
\newtheorem{remark}[theorem]{Remark}

\DeclareMathOperator{\E}{\mathbb{E}}
\DeclareMathOperator{\F}{\mathbb{F}}

\DeclareMathOperator{\N}{\mathbb{N}}

\DeclareMathOperator{\Q}{\mathbb{Q}}
\DeclareMathOperator{\R}{\mathbb{R}}

\DeclareMathOperator{\Z}{\mathbb{Z}}

\DeclareMathOperator{\calC}{\mathcal{C}}

\DeclareMathOperator{\calE}{\mathcal{E}}
\DeclareMathOperator{\calF}{\mathcal{F}}

\DeclareMathOperator{\calH}{\mathcal{H}}

\DeclareMathOperator{\calM}{\mathcal{M}}
\DeclareMathOperator{\calN}{\mathcal{N}}
\DeclareMathOperator{\calO}{\mathcal{O}}

\DeclareMathOperator{\calR}{\mathcal{R}}
\DeclareMathOperator{\calS}{\mathcal{S}}
\DeclareMathOperator{\calT}{\mathcal{T}}

\DeclareMathOperator{\frakF}{\mathfrak{F}}

\providecommand{\keywords}[1]
{
  \small	
  \textbf{\textit{Keywords ---}} \textit{#1}
}

\usepackage{authblk}
\makeatletter
\renewcommand\AB@affilsepx{; \protect\Affilfont}
\makeatother

\renewcommand\Affilfont{\small}

\usepackage[symbol]{footmisc}
\renewcommand{\thefootnote}{\fnsymbol{footnote}}

\usepackage{tikz}
\usetikzlibrary{matrix}

\usepackage{fancyhdr}
\fancypagestyle{firststyle}
{
   \fancyhf{}
   
   \fancyfoot[C]{\scriptsize\thepage}
}

\date{\vspace{-1em}\normalsize{\today}}

\begin{document}

\thispagestyle{firststyle}
\begin{center}
    \Large \textbf{Projection of  
    Functionals and Fast Pricing 
    of Exotic Options
    \footnote[1]{The research has been carried out during Valentin's
    internship 
    at Bloomberg LP. The author wishes to thank  Bruno Dupire for his precious guidance, goodwill and inspiring discussions.  \vspace{2mm}}} \\
    \vspace{0.5cm}
    \normalsize  Valentin Tissot-Daguette\footnote[2]{Operations Research and Financial Engineering   Department, Princeton University, Princeton, NJ 08544, USA. Email: \texttt{v.tissot-daguette@princeton.edu}\vspace{2mm}}\\
\vspace{.1cm}
\end{center}

\vspace{-2mm}
\begin{abstract}
  We investigate the approximation of path  functionals. 
     In particular, we advocate the use of  the Karhunen-Loève expansion, the continuous analogue of Principal Component Analysis, to  extract relevant information from the image of a functional. 
Having accurate estimate of functionals is of paramount importance in the context of exotic derivatives pricing, as presented in the practical applications.
Specifically, 
we show how a simulation-based procedure, 
which we call the Karhunen-Loève Monte Carlo (KLMC) algorithm, 
allows fast and  efficient computation of the price of path-dependent options. We also  explore the path signature as an alternative tool to project both paths and functionals.
\end{abstract}

\\



\renewcommand{\thefootnote}{\arabic{footnote}}
\section{Introduction}

The pricing of exotic 
options remains a difficult task in quantitative finance. The main challenge is to find an adequate trade-off between pricing accuracy and fast computation. Efficient techniques such as finite difference \cite{Schwartz} or the fast Fourier transform \cite{Carr} are in general not applicable to path-dependent payoffs. Practitioners are often forced to turn to  standard Monte Carlo methods, despite being  slow. 
Therefore, researchers have come up with novel ideas over the years to tackle this issue.   
Recently, some authors have employed deep learning to price vanilla and exotic options in a non-parametric manner \cite{Hull} while others
showed the benefits of the path signature  \cite{Szpruch,LyonsNum} to project exotic payoffs in a nonlinear way. 

In this paper, we move away from  prevailing machine learning methods and bring a classical tool  back into play: the Karhunen-Loève (KL) expansion \cite{Karhunen, Loeve}. The latter provides an  orthogonal decomposition of stochastic processes that is optimal in the $L^2(\Q \otimes dt)$ sense. The theory has been  applied to Gaussian processes \cite{Ghanem,Solin},  functional quantization \cite{Pages},  and  recently to the Brownian Bridge in a weighted Hilbert space \cite{Foster}. 
In this paper, the KL expansion takes on a newfound importance when it is applied to the projection of path functionals. We propose a simple simulation-based procedure, which we call the Karhunen-Loève Monte Carlo (KLMC) algorithm, to  compute the price  of exotic options.\footnote{See \href{https://github.com/valentintissot/KLMC.git}{https://github.com/valentintissot/KLMC} for an implementation.}
The superiority of KL-based Monte Carlo methods compared to the ordinary one was shown numerically in \cite{Acworth} for  the Brownian case. Our goal is here to further support these findings and 
extend this approach. We also discuss alternative methods employing the path signature as basis of the space of functionals. 




The remainder of this paper is structured as follows. In \Cref{sec:pathApprox}, we gather standard results from approximation theory and draw a 
parallel between Hilbert projection and the à la mode path signature. 
\Cref{sec:funcApprox} is devoted to the approximation of functionals, where two routes are contrasted as well as a short discussion on the use of the signature for this task. 
We  apply the developed tools and finally present the KLMC algorithm with accompanying numerical evidence in \Cref{sec:application}.
 \section{Path Approximation}
\label{sec:pathApprox}

For fixed horizon $T>0$, let $\Lambda_t = \calC([0,t],\R)$ and $\Lambda := \bigcup_{t\in[0,T]}\Lambda_t$.   For $X \in \Lambda_t$ and $s \le t$, $X_s$ denotes the  trajectory up to time $s$, while $x_s= X(s)$ is the value at time $s$. 
We equip $\Lambda$ with a $\sigma-$algebra $\calF$, filtration $\F$ and probability measure $\Q$ to form a stochastic basis $(\Lambda,\calF,\F, \Q)$. 

The goal of this paper is to price exotic options with payoff of the form $\varphi = h\circ f$, where  $f:\Lambda \to \R$ is a  \textit{functional} and  $h$ a real function. 
 For instance, an Asian call option is obtained with $f(X_t) = \frac{1}{t}\int_{0}^t x_s ds$ and $h(y) = (y-K)^+$. If $\Q$ is a risk-neutral measure and assuming zero interest rate, then 
 $$p = \E^{\Q}[\varphi(X_\tau)],$$ is the value of the  option with payoff $\varphi$ and maturity $\tau \in [0,T]$.  
To compute $p$ using Monte Carlo, we  typically simulate an approximated version of $X \in \Lambda_\tau$, e.g. using time discretization. Alternatively, we can approximate the  \textit{transformed path},
$$Y = f(X), \quad y_t = f(X_t), \quad  t\in [0,\tau],$$ and  write $p = \E^{\Q}[h(y_\tau)]$. We favor the second option, as shown throughout the paper and in the numerical experiments. 
 
 A natural way to approximate  $X$ (or $Y$) is to project it onto a Hilbert space. For simplicity, we focus on paths defined on the whole interval $[0,T]$, so working on $\Lambda_T$ is enough. Also, $f$ is assumed to preserve continuity, so $f(X)\in \Lambda_T$ as well.  We now  present the theory  for the original path, although the same would hold for the transformed one; see \cref{sec:funcApprox}.  
Let $\calH $ be a separable Hilbert space with inner product $(\cdot,\cdot)_{\calH}$. Then any  $X \in  \Lambda_T \cap \,  \calH $ admits the representation
\begin{equation}\label{eq:proj}
    x_t = \sum_{k} \xi_k F_k(t), \quad \xi_k = (X,F_k)_{\calH}, \quad t\in [0,T], 
\end{equation}
\vspace{-4mm}

where $\mathfrak{F} := (F_k)$ is an orthonormal basis (ONB) of $\calH$.\footnote{
The enumeration of $\mathfrak{F}$ will depend on its construction and
common notations. For instance, $\mathfrak{F}$ may or may not include an initial element $F_0$. For fairness sake, however, we always compare projections involving the same number of basis functions.} An  approximation of $X$ is obtained by truncating the  series in $\eqref{eq:proj}$, that is
$x^{K,\frakF}_t = \sum_{k\, \le\, K} \xi_k F_k(t).$ 
Each pair $(K,\frakF)$ thus induces a 
 projection map $\pi^{K,\frakF}:\calH \to \calH$ given by $\pi^{K,\frakF}(X)  = X^{K,\frakF}$. Although paths are assumed to be  one-dimensional for simplicity, the present framework is  easily generalized. Indeed, if $x_t = (x^1_t,\ldots,x^d_t) \in \R^d$,  it suffices to project each component separately, i.e. $x_t^{i,K,\frakF^i} =\sum_{k\, \le\, K} \xi^i_k F^i_k(t)$ with $\xi^i_k = (X^i,F^i_k)_{\calH}$ and $(\frakF^i)_{i=1}^{d}$ ONB's of $\calH$.
 
\subsection{Karhunen-Loève Expansion}\label{ssec:KL}
Let $\calH$ be the Lebesgue space $L^2([0,T])$ of square-integrable functions, where  we write 
$(\cdot,\cdot) = (\cdot,\cdot)_{L^2([0,T])} $ for brevity. 
Among the myriad of bases available, which one should be picked? The answer will depend upon the optimality criterion. One possibility is to minimize the square of the $L^2(\Q \otimes \, dt)-$norm (denoted by $\lVert \cdot \rVert_{*}$) between a path and its $K-$order  truncation, i.e. 
$$\epsilon^{K,\frakF} := \lVert X - X^{K,\frakF}\rVert^2_{*} = \E^{\Q} \int_0^T |
x_t - x^{K,\frakF}_t|^2 dt,$$
for $X \in \Lambda_T \cap L^2([0,T])$. 
Thanks to the orthogonality of $\frakF$,  we have 
\begin{equation}\label{eq:err}
   \epsilon^{K,\frakF} = \sum_{k,l \, > \, K}(\xi_k,\xi_l)_{L^2(\mathbb{Q})} \;(F_k,F_l)_{L^2([0,T])}  = \sum_{k \, > \, K} \lambda^{\frakF}_k , \quad \lambda^{\frakF}_k :=  \lVert  \xi_k\rVert^2_{L^2(\Q)}, 
\end{equation}
where it is assumed that $\lambda_k^{\frakF} \ge \lambda_l^{\frakF}\; \; \forall \ k < l $ without loss of generality. 
As  the mapping 
$\frakF \mapsto \sum_{k} \lambda^{\frakF}_k  $ is  constant and equal to the total variance $\lVert X\rVert^2_{L^2(\Q \otimes \, dt)}$, the projection error is solely determined by the speed of decay of  $(\lambda^{\frakF }_k)$. Inversely, the optimal basis will maximize the cumulative sum of  variance $\sum_{k \, \le \, K} \lambda^{\frakF}_k$.
This leads us to the \textit{Karhunen-Loève expansion} \cite{Karhunen,Loeve}, the continuous analogue of Principal Component Analysis. 
In what follows, assume  $\E^{\Q}[x_t]=0$ $\forall \, t \in [0,T]$ and define the covariance kernel $\kappa_X(s,t) = (x_s, x_t)_{L^2(\Q)}.$ 
As $\kappa_X$ is symmetric,  continuous and non-negative definite, Mercer's representation theorem \cite{Mercer} ensures the existence of  an ONB $\frakF=(F_k)$ of $L^2([0,T])$ and scalars $\lambda_1^{\frakF} \ge \lambda_2^{\frakF} \ge \ldots \ge 0$ such that 
\begin{equation}\label{eq:mercer}
    \kappa_X(s,t)= \sum_{k=1}^{\infty} \lambda^{\frakF}_k F_k(s) F_k(t).
\end{equation}
Then $\frakF$ is the \textit{Karhunen-Loève (KL) basis} associated to $X$ under $\Q$. 
From $\eqref{eq:mercer}$, it is immediate  that   $F_k$ solves the Fredholm integral equation
 $$(\kappa_X(t,\cdot),F_k) = \lambda_k^{\frakF}\,  F_k(t), \quad  t\in [0,T].$$ Accordingly, 
 $\frakF$ and $(\lambda_k^{\frakF})$ are  termed \textit{eigenfunctions} and \textit{eigenvalues} of $\kappa_X$, respectively.  
Observe that the squared $L^2(\Q)$ norm of the KL coefficient $\xi_k=(X,F_k)$ is precisely $\lambda_k^{\frakF}$, whence comes the notation in $\eqref{eq:err}$. 
The next result reflects the relevance of the KL expansion; see \cite[Theorem 2.1.2.]{Ghanem} for a proof.

\begin{theorem}
\label{thm:KL}
The Karhunen-Loève basis
is the unique ONB of $L^2([0,T])$ minimizing $\epsilon^{K,\frakF}$  for every truncation level $K\ge 1$. 
\end{theorem}

\begin{remark}\label{rem:center}
For non-centered trajectories,  it suffices to characterize the Karhunen-Loève basis of $x_t-\E^{\Q}[x_t]$. The projected path is then obtained by  adding the mean function back to the expansion. 
\end{remark}

\begin{example} \label{ex:KLBM} Let $T=1$ and $\Q$ be the Wiener measure, i.e. the coordinate process $X$ is Brownian motion on $[0,1]$. 
The covariance kernel is  $\kappa_X(s,t) = s \wedge t$, leading  to the eigenfunctions $F_k(t) = \sqrt{2} \sin((k-1/2)\pi t)$ and eigenvalues $\lambda_k^{\frakF} = \frac{1}{\pi^2(k-1/2)^2}$, $k \ge 1.$  
The projection error is approximately equal to 
\begin{align*}
    \epsilon^{K,\frakF}
    = \frac{1}{\pi^2}\sum_{k\,>\,K} \frac{1}{(k-1/2)^2}
    \approx \frac{1}{\pi^2} \int_K^{\infty} \frac{dk}{(k-1/2)^2} = 
    \frac{1}{\pi^2(K-1/2)}.
\end{align*}
It is easily seen that $\xi_k = (X,F_k) \sim \calN(0,\lambda_k^{\frakF})$ so that $\xi_k, \xi_l$ are independent for $k\ne l$.  Therefore, "smooth" Brownian motions can be simulated  by setting
$x^{K,\frakF}_t = \sum_{k =1}^K \sqrt{\lambda_k^{\frakF}}\, z_k \, F_k(t)$ with  $(z_k)_{k=1}^K \overset{\textnormal{i.i.d.}}{\sim} \calN(0,1)$, $K\ge 1.$
\end{example}

\subsection{Lévy-Cieselski Construction}\label{ssec:LC}
Another Hilbert space  is the \textit{Cameron–Martin space}, 
$\calR = \{F \in \Lambda_T \, | \, dF \ll dt, \, \dot{F} \in L^2([0,T]) \} $ 
where $\dot{F}$ denotes the time derivative of $F$. The inner product is  $(F,G)_{\mathcal{R}} = (\dot{F}, \dot{G})$, from which 
$(F_k)$  is an ONB of  $\calR \; \Longleftrightarrow \; (\dot{F}_k)$  is an ONB of  $L^2([0,T])$
is immediate.
If $X^{K,\frakF}$ is a projected path with respect to an ONB $\frakF$ of $\calR$, then taking derivative gives
$$\dot{x}_t^{K,\frakF} = \sum_{k \, \le \, K} (\dot{X},\dot{F}_k)\dot{F}_k(t) = \sum_{k \, \le \, K} (X,F_k)_{\calR}\, \dot{F}_k(t). $$
 We gather that the projection of a path onto $\calR$ corresponds to an  $L^2([0,T])$ projection of its (possibly generalized)  derivative 
 followed by a time integration.
When $\Q$ is the Wiener measure, this procedure  is often called the \textit{Lévy-Cieselski construction}.  
With regards to accuracy, we recall the expression for the  $L^2(\Q \otimes \, dt)-$error,
$$\epsilon^{K,\frakF} = \lVert X - X^{K,\frakF}\rVert^2_{*} = \sum_{k,l \, > \, K}(\xi_k,\xi_l)_{L^2(\mathbb{Q})} \;(F_k,F_l)_{L^2([0,T])}.\vspace{-2mm}$$
As orthogonal functions in $\calR$ need not be orthogonal in $L^2([0,T])$, we cannot in general get rid of the double sum above. 
However, if $\Q$ is the Wiener measure and $\dot{X}$ the white noise process, then Fubini's theorem gives 
$$(\xi_k,\xi_l)_{L^2(\mathbb{Q})} = \int_{[0,T]^2}  \underbrace{\E^{\Q}[\dot{x}_s \dot{x}_t]}_{=\, \delta(t-s)}\dot{F}_k(s) \dot{F}_l(t) ds dt = \int_0^1 \dot{F}_k(t) \dot{F}_l(t) dt = (F_k,F_l)_{\calR}= \delta_{kl} .$$
Thus, 
$\epsilon^{K,\frakF} = \sum_{k\, >\, K} \lVert F_k \rVert^2 $.
The optimal Cameron-Martin basis would therefore have the fastest decay of its squared norms $(\lVert F_k \rVert^2)$, assuming the latter are sorted in non-increasing order.
We illustrate the Lévy-Cieselski construction with two examples, taking $T=1$.

\begin{example}\label{ex:BBC}

A standard method to prove the existence of Brownian motion follows from the \text{Brownian bridge construction}. In short, it consists of a random superposition of triangular functions$-$the Schauder functions$-$obtained by integrating the Haar basis on $[0,1]$,
$$\dot{F}_{k,l}(t)=2^{k/2} \, \psi \left(2^{k}t-l\right),\quad 0 \le l < 2^k,\quad t\in [0,1],$$
with the wavelet $\psi= (-1)^{ \mathds{1}_{[1/2, 1)}}$,  $\textnormal{supp}(\psi) = [0,1]$. The Schauder and Haar functions are illustrated on the left side of  \Cref{fig:CMS}.  It is easily seen that $\dot{F}_{k,l}$ as well as $F_{k,l}$ have support  
$[l/2^k,(l+1)/2^k]$, the $l-$th subinterval of the dyadic partition $\Pi_{k} = \{l/2^k\,|\, 0 \le l \le  2^k\}$. 
The construction  is  incremental:   
First, the terminal value of the path is simulated. Then, for each subinterval of $\Pi_k$, $k\ge 0$,  a random value for the mid-point  is generated  and thereafter  connected  to the endpoints in a linear fashion  (using  $F_{k+1,\cdot}$).  
The restriction of $X$ to $\Pi_k$ will therefore remain the same when  finer characteristics of the path are added.

When considering all functions up to the $\bar{K}-$th dyadic partition, the total number of basis functions employed  is $K = 
2^{\bar{K}+1}-1$.   
For Brownian motion, the approximation error is known \cite{Brown} and equal to  
$\epsilon^{K,\frakF} = \frac{1}{6K}.$

\end{example}

\begin{example}\label{ex:cosine}
Let $(\dot{F})$ be the \text{cosine Fourier ONB}, i.e.
$\dot{F}_k(t) = \sqrt{2}\cos(\pi k t)$, $t \in [0,1]$. The anti-derivatives $F_k(t) = \sqrt{2}\,\frac{\sin(\pi k t)}{\pi k}$ turns out to correspond$-$up to a factor$-$to the Karhunen-Lo\`eve basis of the \text{Brownian bridge}. Indeed, recalling that $\kappa_X(s,t) = s\wedge t - st$ if $X$ is a Brownian bridge, we have for the ONB $\, \tilde{\frakF} = (\tilde{F}_k) = (\pi k \, F_k)$, 
\begin{align*}
    (\kappa_X(\cdot,t),\tilde{F}_k) 
    = \sqrt{2}\left[(1-t)\int_0^t s\, \sin(\pi k s) ds + t \int_t^1 (1-s)\sin(\pi k s) ds\right]
    = \sqrt{2}\, \frac{\sin(\pi k t)}{\pi^2 k^2},
\end{align*}
using integration by parts in the last equality. 
The eigenvalues  are therefore $(\lambda_k^{\tilde{\frakF}}) = (\frac{1}{\pi^2 k^2})$. 
The first elements of $\tilde{\frakF}$ and the Fourier cosine ONB are displayed on the right charts of  \Cref{fig:CMS}.
Following the same argument as in  \cref{ex:KLBM}, the (minimal) projection error onto $K$ basis functions is approximately equal to $\frac{1}{\pi^2 K}$. This is less than Brownian motion (see \Cref{ex:KLBM}) as little more is known about a Brownian bridge; $\Q-$almost all trajectories return to the origin.

\end{example}

 \begin{figure}[H]
    \centering
     \caption{Basis functions and derivatives in the Cameron-Martin space.}
     \vspace{-2mm}
    \includegraphics[scale = 0.33]{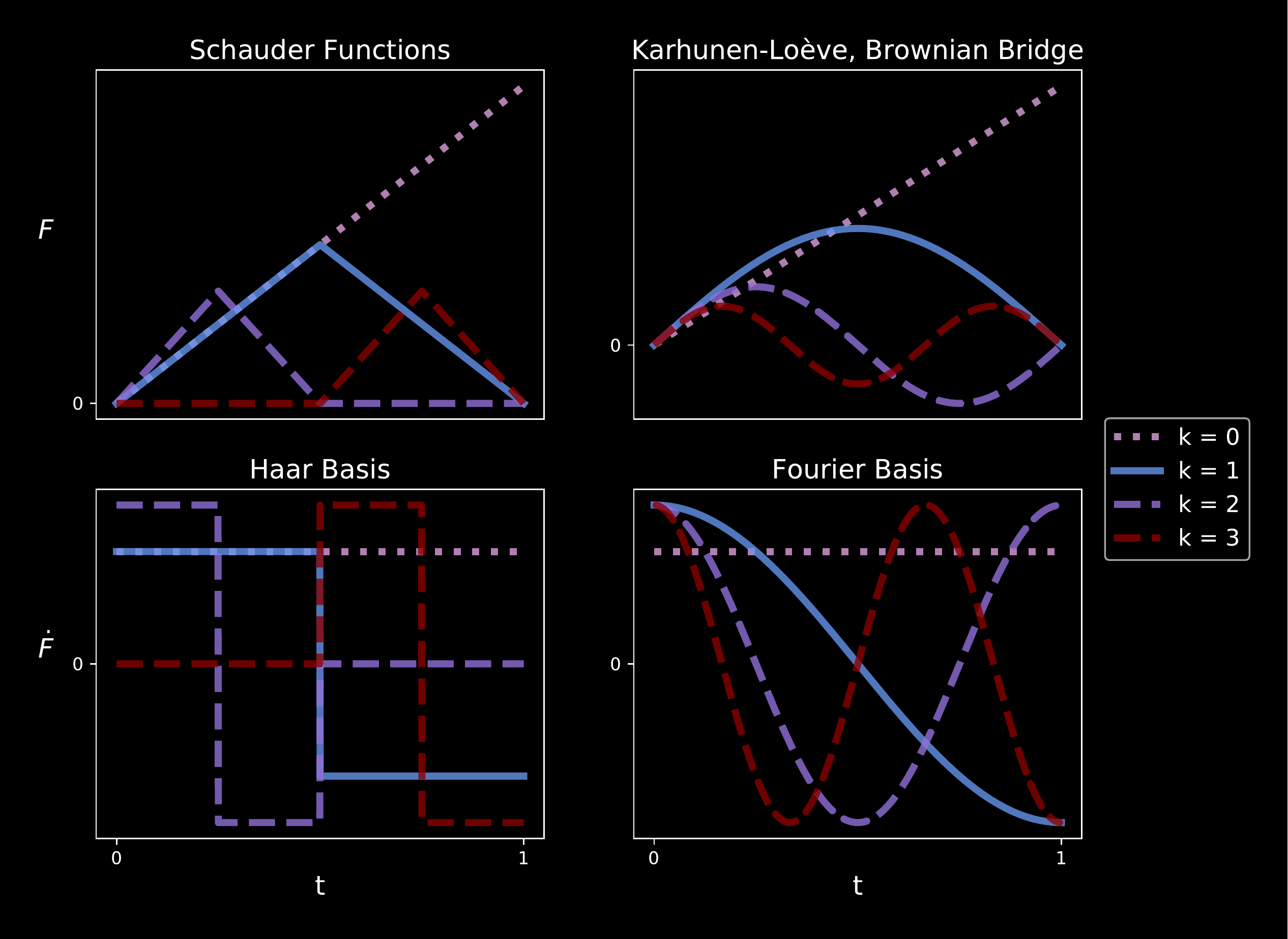}
    \label{fig:CMS}
\end{figure}

\subsection{Signature and Legendre Polynomials}\label{sec:sigLegendre}
An alternative characterization of a path is available through the so-called \textit{signature}  \cite{Lyons}.  Roughly speaking, the signature extract  from a path an infinite-dimensional skeleton, where  each "bone" contains inherent information. 
We start off with a few definitions. 
A \textit{word} is a sequence  $\alpha = \alpha_1 ... \alpha_k$ of letters from the alphabet $\{0,1\}$. The length of $\alpha$ is denoted  by $l(\alpha)$. 
Moreover, we augment a path $X \in \Lambda$ with the time itself $t \mapsto t$  and   write 
$x^0_{t} = t$, $x^1_{t} = x_t$.  The words $0,1$ are therefore identified with the time $t$ and path $x$, respectively. 
\begin{definition}
The \textit{signature}  is   a collection  of functionals $\calS= \{\calS_{\alpha}: \Lambda \to \R\}$ given by 
\begin{equation}\label{eq:sig}
   \calS_{\emptyset} \equiv 1, \qquad  \calS_{\alpha}(X_t) =
\int_{0}^{t} \int_{0}^{t_k} \cdots \int_{0}^{t_2} \circ \, dx^{\alpha_1}_{t_1} \cdots \circ dx^{\alpha_k}_{t_k}, \quad l(\alpha)=k \ge 1,
\end{equation}
     where  $\circ$ indicates Stratonovich integration.\footnote{The integrals in  $\eqref{eq:sig}$ are in the Lebesgue-Stieltjes  (respectively It\^o) sense when  the integrator (respectively integrand) is of bounded variation. In both cases, the symbol $\circ$ can be omitted.\\[-0.5em]}
\end{definition} 
When referring to a specific path $X$, we shall call the sequence $(\calS_{\alpha}(X))$  the \textit{signature of $X$}. This is usually how the signature is defined; see  \cite{Lyons}. 
 If $x_t\in \R^d$ with $d\ge 2$, then the alphabet becomes $\{0,1,\ldots,d\}$ and  the signature is  obtained  analogously.  
The first signature functionals read
\begin{align*}
    \calS_{0}(X_t) &= \int_0^t d t_1 = t, &&\calS_{1}(X_t) = \int_0^t \circ \,d x_{t_1} = x_t - x_0,\\
  \calS_{00}(X_t) &= \int_0^t \int_0^{t_2} d t_1 d t_2 =\frac{t^2}{2}, &&\calS_{11}(X_t) = \int_0^t \int_0^{t_2}\circ \, d x_{t_1} \circ d x_{t_2}=\frac{(x_t-x_0)^2}{2},\\
  \calS_{10}(X_t) &= \int_0^t \int_0^{t_2} dx_{t_1}  d t_2 =\int_0^t (x_{s} - x_0) ds, &&\calS_{01}(X_t) = \int_0^t \int_0^{t_2} d t_1 d x_{t_2} =\int_0^t s \,dx_s.
\end{align*}
Keeping track of the passage of time  is crucial, as the signature would otherwise barely carry  information about the path. Indeed, notice that
$\calS_{\alpha}(X_t) = \frac{(x_t-x_0)^{k}}{k!}$ for  $\alpha =  1...1$, $l(\alpha)=k$ 
(as seen above for $k=1,2$) thus only the increment $x_t-x_0$ is known with the alphabet $\{1\}$. 
 
A property of the signature is that it uniquely characterizes a path, up to a  equivalence relation: two paths having same signature differ at most by a \textit{tree-like path} \cite{Hambly}, a specific type of loop.  
Hence, extending a path with time 
not only enriches the  signature  but also
ensures  injectivity  as $t\to x^0_t =t$ is increasing. 
This gives hope to reconstruct the (unique) path associated to a  signature sequence. This was investigated 
 in \cite{Geng}, where the author shows a 
geometric reconstruction using polygonal approximations for Brownian paths. 
We propose a simple algorithm, in connection with our discussion on Hilbert projections.\footnote{I thank Bruno Dupire for suggesting this interesting parallel.} 
For ease of presentation, assume $x_0=0$ and $T=1$. We first make the following observation. 

\begin{lemma}\label{lem:Legendre}
Let $\overleftarrow{X}$ denote the \textit{time reversed path}, i.e. $\overleftarrow{\ x_t} = x_{1-t}$ and 
introduce the words $\alpha^{(k)} :=10\ldots0\,$, $l(\alpha^{(k)})=k+2$, $k \ge 0$. 
Then $\calS_{\alpha^{(k)}} (\overleftarrow{X}_{\! 1})  = \frac{1}{k!}(X, m_k) $ where $m_k(t)=t^k$.
\end{lemma}
\begin{proof} First, observe that 
\begin{equation}\label{eq:100}
    \calS_{\alpha^{(k)}}(X_{t}) = \int_0^{t} x_s \frac{(t -s)^{k}}{k!}ds,\quad \forall \, t \in [0,1]. 
\end{equation}
Indeed for fixed $t\in [0,1]$ and  $k=0$, then $\calS_{\alpha^{(0)}}(X_t)=\calS_{10}(X_t) = \int_0^t x_sds$, which is $\eqref{eq:100}$. Now by induction on $k\ge 1$, uniformly on $[0,t]$,
\begin{align*}
  \calS_{\alpha^{(k)}}(X_t) = \int_0^t \calS_{\alpha^{(k-1)}}(X_u)du 
    = \int_0^t \int_0^u x_s \frac{(u -s)^{k-1}}{(k-1)!}ds\, du 
    = \int_0^t x_s \frac{(t -s)^{k}}{k!} ds.
\end{align*}
Now taking $t=1$ and $\overleftarrow{X}$ instead of $X$, we get
$
    \calS_{\alpha^{(k)}}(\overleftarrow{X}_{\! 1}) = \int_0^1 \overleftarrow{\ x_t}\frac{(1 -t)^{k}}{k!}dt =  \frac{1}{k!}(x,m_k).
$
\end{proof}

Essentially,  \Cref{lem:Legendre} states that the knowledge of $(\calS_{\alpha^{(k)}}(\overleftarrow{X}_{\! 1}))_{k\ge 0}$ is equivalent to the knowledge of the inner products of the path with the monomials. Since  also 
\begin{align*}
    \calS_{\alpha^{(k)}}(\overleftarrow{X}_{\! 1})
    &= (-1)^{k}\int_0^1 x_t \frac{((1-t)-1)^{k}}{k!}dt 
    = \sum_{j=0}^{k} \calS_{\alpha^{(j)}}(X_1) \frac{(-1)^{j}}{(k-j)!},
\end{align*}
the coefficients $(X,m_k)_{k\ge 0}$ can be retrieved from $(\calS_{\alpha^{(k)}}(X_1))_{k\ge 0}$ as well. 
To fall within the context of orthonormal projection, we transform the monomials into the (unique) polynomial ONB of $L^2([0,1])$. 
Let $(p_k)$ be the Legendre polynomials \cite{Szego}, forming an orthogonal basis of $L^2([-1,1])$. Then  consider the \textit{shifted Legendre polynomials},  $q_k(t) = p_k(2t-1),$  $t\in [0,1].$ 
We write 
$$q_k(t)= \sum_{j\le k} a_{k,j}t^j, \quad a_{k,j} = (-1)^{k+j} {k \choose j} {k + j \choose j},$$ with coefficients  derived for instance from  \textit{Rodrigues' formula} \cite[Section 4.3]{Szego}. The standardization $F_k :=  \frac{q_k}{\lVert q_k\rVert} = \sqrt{2k+1}q_k$ makes $\frakF = (F_k)$ an ONB of $L^2([0,1])$. This leads us to the following result. 
\begin{proposition}
If  $b_{k,j} := \sqrt{2k+1} j!\, a_{k,j}$ and $G_j(t) := \sum_{k= j}^K b_{k,j}\,  F_k(t)$, then 
\begin{equation}\label{eq:sigLeg}
    x^{K,\frakF}_t 
    = \sum_{k\le K} \xi_k F_k(t)
    =  \sum_{j \le K} \calS_{\alpha^{(j)}}(\overleftarrow{X_1})  G_j(t).
\end{equation}
\end{proposition}

\begin{proof}
 First,  $(X,F_k) = \sum_{j\le k} a_{k,j}\, (X,m_j)$ with  $m_j(t) = t^j$.   Using \cref{lem:Legendre}, we obtain 
$$\xi_k = \sum_{j \le k} b_{k,j}  \calS_{\alpha^{(j)}}(\overleftarrow{X_1}).$$ Thus   $ X^{K,\frakF}   = \sum_{j \le K} \calS_{\alpha^{(j)}}(\overleftarrow{X_1}) \  G_j$  
with $G_j$ as in the statement.
\end{proof}
In summary, the signature elements $(\calS_{\alpha^{(k)}})_{k\le K}$ generate the $L^2([0,1])$ products of the path with the monomials$-$and in turn, with the Legendre polynomials$-$from which the projected path $X^{K,\frakF} $ becomes available. 
We can therefore retrieve $X$ by letting $K\to \infty$. Note that this reconstruction  works for multidimensional paths as well, as we can apply the procedure to each component  $i=1,\ldots,d$ with the words
$\alpha^{(i,k)} :=i 0\ldots0$, $l(\alpha^{(i,k)}) = k+2$, $k \ge 0$.  
We finish this section by computing the projection error of $\eqref{eq:sigLeg}$ when $X$ is Brownian motion. Recalling that $\kappa_X(s,t) = s \wedge t$,  a  simple calculation gives 
$$\lambda_k^{\frakF} = \E^{\Q}[\xi_k^2] = 2 \int_{0}^1 \int_{0}^t s  F_k(s)dsF_k(t)dt = 2 (2k+1)\sum_{i,j \le k} \frac{a_{k,j}\ a_{k,i}}{(j+2)(i+j+3)}.
$$
 
The first values are given by $(\lambda^{\frakF}_k)_{k=0}^3 = (\frac{1}{3},\frac{1}{10}, \frac{1}{42}, \frac{1}{90})$ from which we  conjecture that $\lambda^{\frakF}_k = \frac{1}{(2k+3)(4k-2)}$ for all $k\ge 1$. This is supported by the fact that $\lambda^{\frakF}_k$ must be  rational numbers as $a_{k,j}\in \Z$ for all $k,j$ and 
$$\sum_{k=0}^K  \lambda^{\frakF}_k = \frac{1}{3} + \frac{1}{8} \sum_{k=1}^K  \left( \frac{1}{2k-1}-\frac{1}{2k+3} \right) = \frac{1}{2} - \frac{K+1}{(2K+3)(4K+2)} \xrightarrow{K \uparrow \infty} \frac{1}{2},$$ coinciding with the total variance of Brownian motion on $[0,1]$.  
Thus, the approximation error reads  $\epsilon^{K, \frakF} =\frac{K+1}{(2K+3)(4K+2)} = \calO(\frac{1}{8K})$, which is of course larger than the Karhunen-Loève basis but smaller than the Brownian Bridge construction (\cref{ex:BBC}). Note that polynomial ONB's may well be optimal if the approximation criterion is modified. In \cite{Foster}, the authors show that in the weighted Hilbert space $L^2([0,1],\mu)$, $\mu(dt) = \frac{dt}{t(1-t)}$, the Karhunen-Loève basis of the Brownian bridge is formed by the anti-derivatives of the Legendre polynomials. Although the construction is different, it  is also curiously related to the signature elements $(\calS_{\alpha^{(k)}})$; see  Theorem 2.3 and  2.4 in \cite{Foster}.


\begin{remark}
Note that the approximation in $\eqref{eq:sigLeg}$ may be improved by adding other elements of the  signature, especially those that are \textit{nonlinear} in $X$, e.g. $\calS_{110}(X_t) = \frac{1}{2}\int_0^t x^2_s ds $.  We postpone this discussion to  \Cref{sec:sigFunc} when projecting  running functionals.
\end{remark}

 \subsection{Numerical Results}\label{sec:numResultX}

We concentrate our experiments on Brownian trajectories.  
First, we illustrate 
the  path  approximations seen earlier (Karhuhen-Loève, L\'evy-Cieselski, Signature).
Figure \ref{fig:projPath} displays the projections using $K=8$ basis elements. We naturally notice similarities between the Karhunen-Loève transform and the L\'evy-Cieselski construction with Fourier cosines, both 
obtained by superposing trigonometric functions.

\begin{figure}[H]
    \centering
    \caption{Projected paths with $K=8\,$ basis elements. }
    \vspace{-2mm}
    \includegraphics[scale = 0.43
    ]{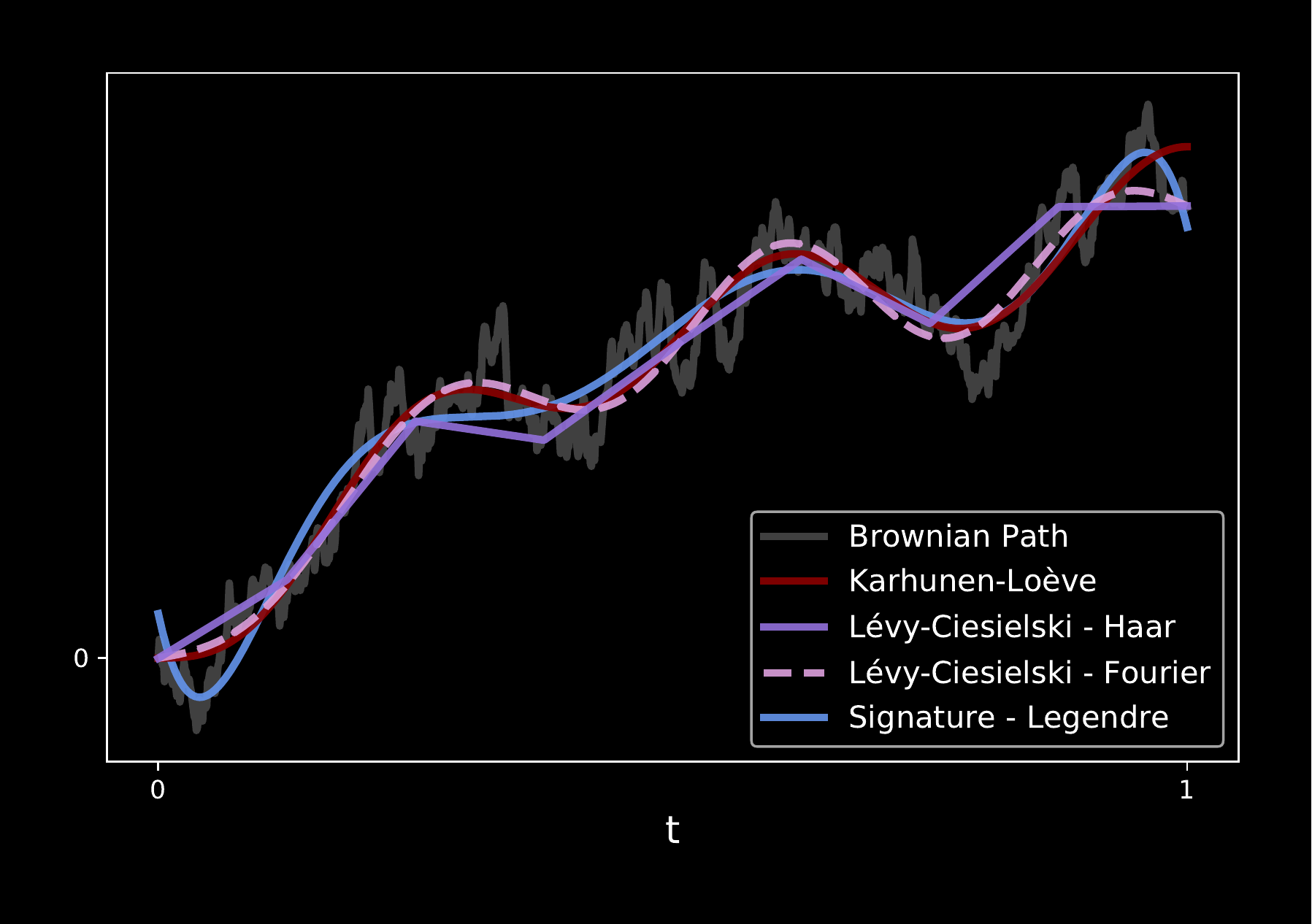}
    \label{fig:projPath}
\end{figure}

Let us gauge the accuracy of the above approximations for Brownian trajectories, in terms of (i) 
    $\epsilon^{K,\frakF}$ and (ii) variance explained 
    $ \vartheta^{K,\frakF} := \frac{\lVert X^{K,\frakF}  \rVert^2_{*}}{\left \lVert X \right \rVert^2_{*}}.$ 
To compute (i), (ii) and the coefficients $(X,F_k)_{\calH}$, we discretize the interval $[0,1]$ a regular partition made of $N = 10^4$ subintervals. 
\Cref{fig:Error_VarExp} displays the evolution of $\epsilon^{K,\frakF}$, $\vartheta^{K,\frakF}$ for $K\in \{1,\ldots,128\}$. The Karhunen-Lo\`eve expansion clearly dominates the other projections, although being asymptotically equivalent to the L\'evy-Cieselski construction with Fourier cosine basis. Besides, the $L^2(\mathbb{Q} \, \otimes \, dt)$ convergence of the Brownian bridge construction (L\'evy-Cieselski with Haar basis) is non-monotonic. Indeed, a bump appears until a full cycle of the dyadic partition is completed. 
Lastly, the slopes in the log-log convergence plot  (left chart of \Cref{fig:Error_VarExp}) 
are roughly equal to $-1$. Put differently,  the squared approximation error is of order $\calO(\frac{1}{K})$, confirming our findings from the  above examples.

\begin{figure}[H]
    \centering
    \caption{$L^2(\mathbb{Q} \, \otimes \, dt)$ error and variance explained.}
    \vspace{-2mm}
    
    \includegraphics[scale = 0.46]{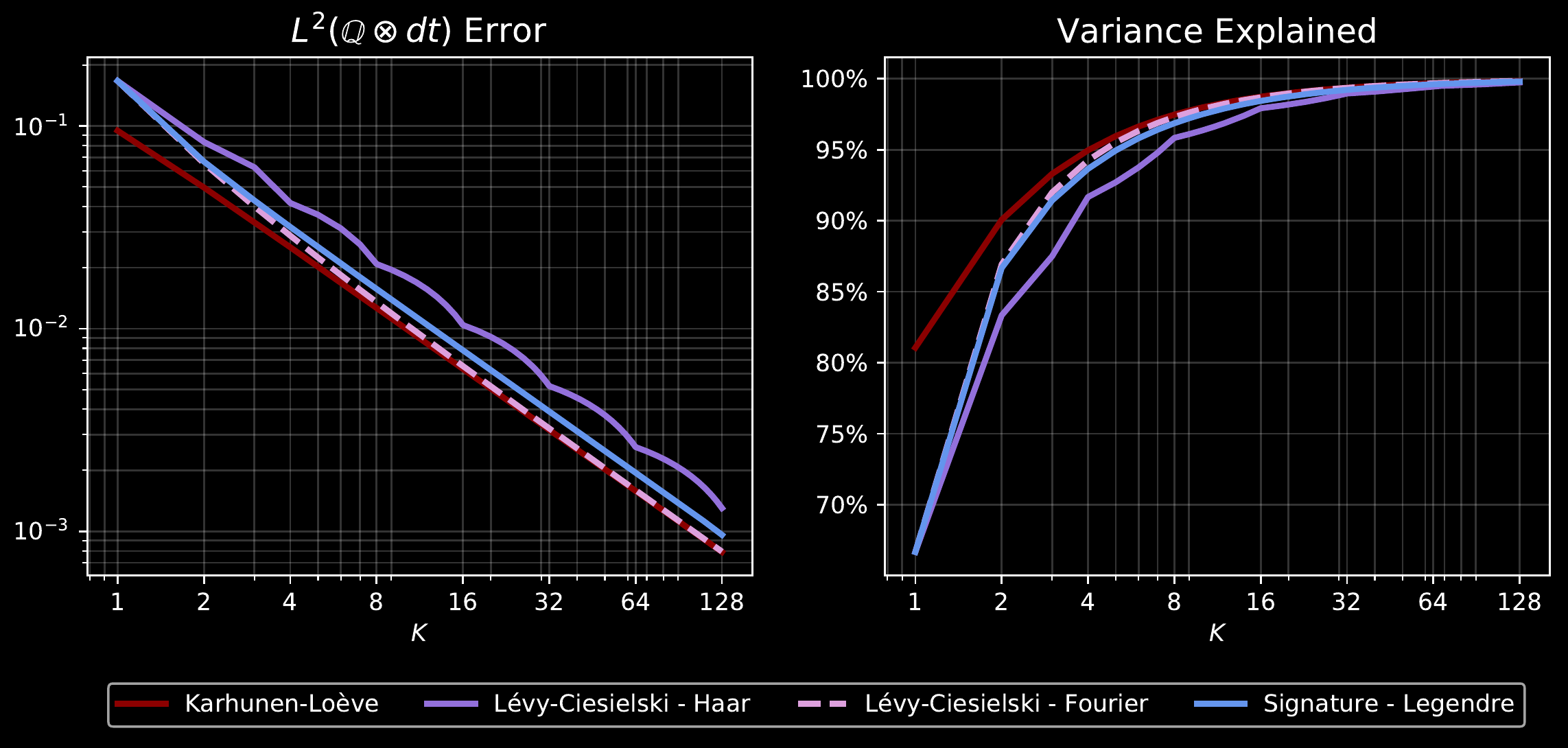}
    \label{fig:Error_VarExp}
\end{figure}

\section{Projection of Functionals}
\label{sec:funcApprox}
In  \Cref{sec:pathApprox}, we unveiled two ways to approximate exotic payoffs $\varphi = h\circ f$, namely by projecting the original path $X$ or $Y = f(X)$ directly.  If $\pi^{K,\frakF}$ denotes as before the projection map onto a Hilbert space $\calH$, we can therefore write $\varphi^{K,\frakF} := h\circ f^{K,\frakF}$, where either $f^{K,\frakF} = f \circ \pi^{K,\frakF}$ (functional of projected path)
or $f^{K,\frakF} = \pi^{K,\frakF} \circ f$ (projected functional). 
We shall see in  \Cref{ssec: numResult} that the former is suboptimal. 
Although not so problematic for functionals capturing global features of a path, local path characteristics (e.g. running maximum) will typically be grossly estimated. Indeed, projecting a path first erases most of its microstructure. 
We thus favor the second option ($f^{K,\frakF} = \pi^{K,\frakF} \circ f$), which consists of replacing $X$ by $Y$ in $\eqref{eq:proj}$. Let us now focus on $\calH = L^2([0,T])$ and demonstrate how to compute the Karhunen-Loève basis of $Y$.

\subsection{Karhunen-Loève Expansion of Functionals}

Assume that $Y \in L^2([0,T]) \cap \Lambda_T$ has  zero mean (otherwise, see \cref{rem:center}).   \cref{thm:KL} suggests to set  $\frakF$  equal to the eigenfunctions of  $\kappa_Y(s,t) = (y_s,y_t)_{L^2(\Q)}$. Optimality comes, however,  at the cost of explicitizing $\frakF$. We proceed as follows: take a regular partition $\Pi_N = \{t_n = n\, \delta t\, |\, n=0,\ldots,N\}$, $ \delta t =\frac{T}{N}$ and compute the kernel matrix $\kappa^{N}_Y = (\kappa_Y(t_n,t_m))_{0 \le n,m \le N}$. When $\kappa_Y$ does not admit a  closed-form expression, $\kappa^{N}_Y$ is replaced by the sample covariance matrix using simulated paths for $Y$. The eigenfunctions thus become eigenvectors and solve the systems\footnote{In  $\eqref{eq:eigendecomp}$,  $\sum"$ means that the first and last summand are halved, i.e. the trapezoidal rule is used to compute  $(\kappa_Y(\cdot,t), F_k)$. Another approach, known as Nyström's method  \cite{reinhardt} consists of   employing a Gaussian quadrature scheme instead.  
However, for large $N$,  we haven't observed any improvement 
and thus favor the more convenient discretization in  $\eqref{eq:eigendecomp}$.}
\vspace{-1mm}
\begin{equation}\label{eq:eigendecomp}
        \sum_{t_n\in  \Pi_N}\!\!" \, \kappa^{N}_{Y}(t_n,t_m) F_{k}(t_n)  \delta t = \lambda^{\frakF}_{k}\,  F_{k}(t_m), \quad t_m\in \Pi_N, \quad k=0,...,N.
\end{equation}
\vspace{-3mm}

 This is a simple eigenvalue problem so  all pairs $(F_k,\lambda_k^{\frakF})$ can be  computed in one go. Let us proceed with two examples where  $T=1$ and $\Q=$ Wiener measure throughout.

\begin{example} \label{ex:timeIntAvg}
 Consider  the time integral and average of a Brownian path, 
$y_t = 
f(X_t) = \int_0^t x_s\, ds, \ \bar{y}_{t}= 
\bar{f}(X_t) = \frac{1}{t}f(X_t). $
These are clearly centered processes and their covariance kernels of can be found explicitly. Starting with $Y$, 
$$\kappa_{Y}(s,t) =  \left(\int_0^s x_r dr,\int_0^t x_u du \right)_{L^2(\Q)} \overset{\textnormal{Fubini}}{=} \int_0^s\int_0^t \kappa_X(r,u) dr du.$$
A straightforward calculation gives
$\kappa_{Y}(s,t)  = \frac{s^2 t}{2} - \frac{s^3}{6}$ 
and 
$\kappa_{\overline{Y}}(s,t)   = \frac{s}{2} -\frac{s^2}{6t},  \, s\le t,$ where $\overline{Y} = \overline{f}(X)$. 
We display in  \Cref{fig:AvgK,fig:IntK} the covariance kernel (top) and first eigenfunctions (bottom) of $f$ and $\bar{f}$, respectively.
The dashed lines in the top panels  are the eigenfunctions of the original (Brownian) path. Note the wider range in the eigenfunctions $F_1,F_2$ for $\bar{f}(X)$ compared to the integrated path for small $t$.  This might come from the greater fluctuations of the time average at inception. 

\vspace{-3mm}

\end{example}

\begin{figure}[H]
\caption{Covariance kernels (top) and eigenfunctions (bottom).}
\vspace{-2mm}
\begin{subfigure}[b]{0.325\textwidth}
    \centering
    \caption{Time average}
    \includegraphics[height=1.4in,width=1.85in]{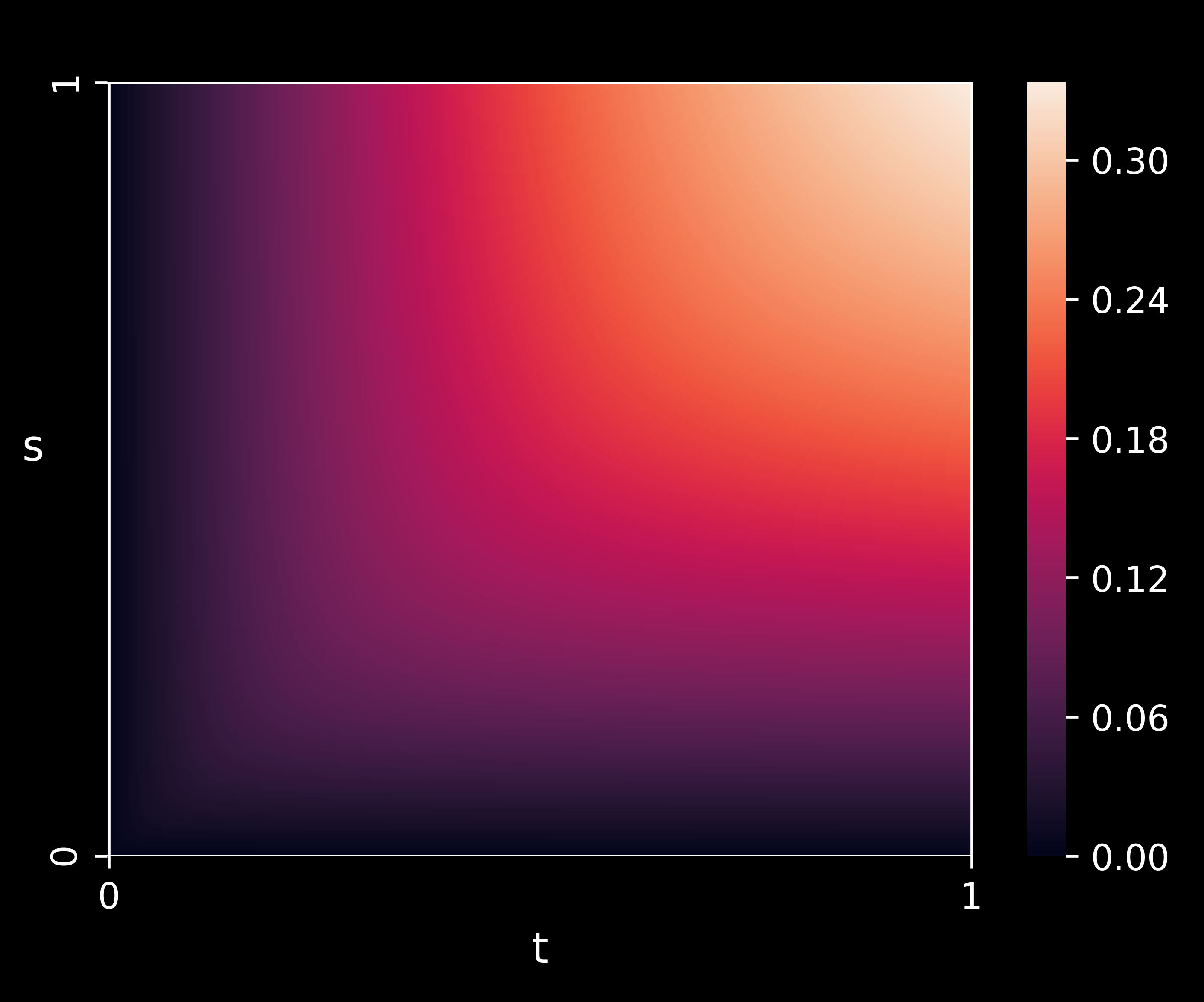}
    \label{fig:AvgK}
\end{subfigure}
\begin{subfigure}[b]{0.325\textwidth}
    \centering
    \caption{Time integral}
    \includegraphics[height=1.4in,width=1.85in]{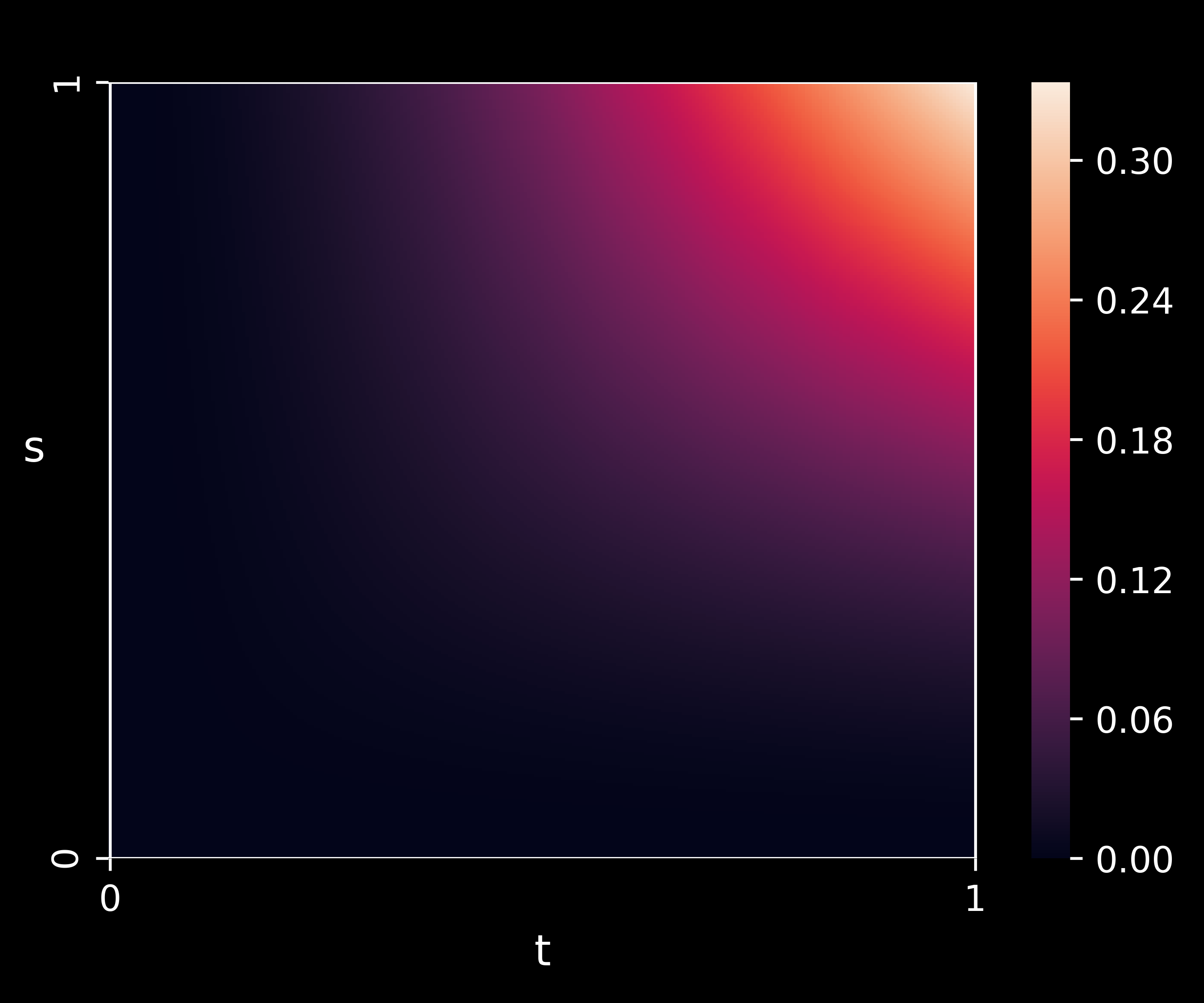}
    \label{fig:IntK}
\end{subfigure}
\begin{subfigure}[b]{0.325\textwidth}
    \centering
    \caption{Running maximum}
    \includegraphics[height=1.4in,width=1.85in]{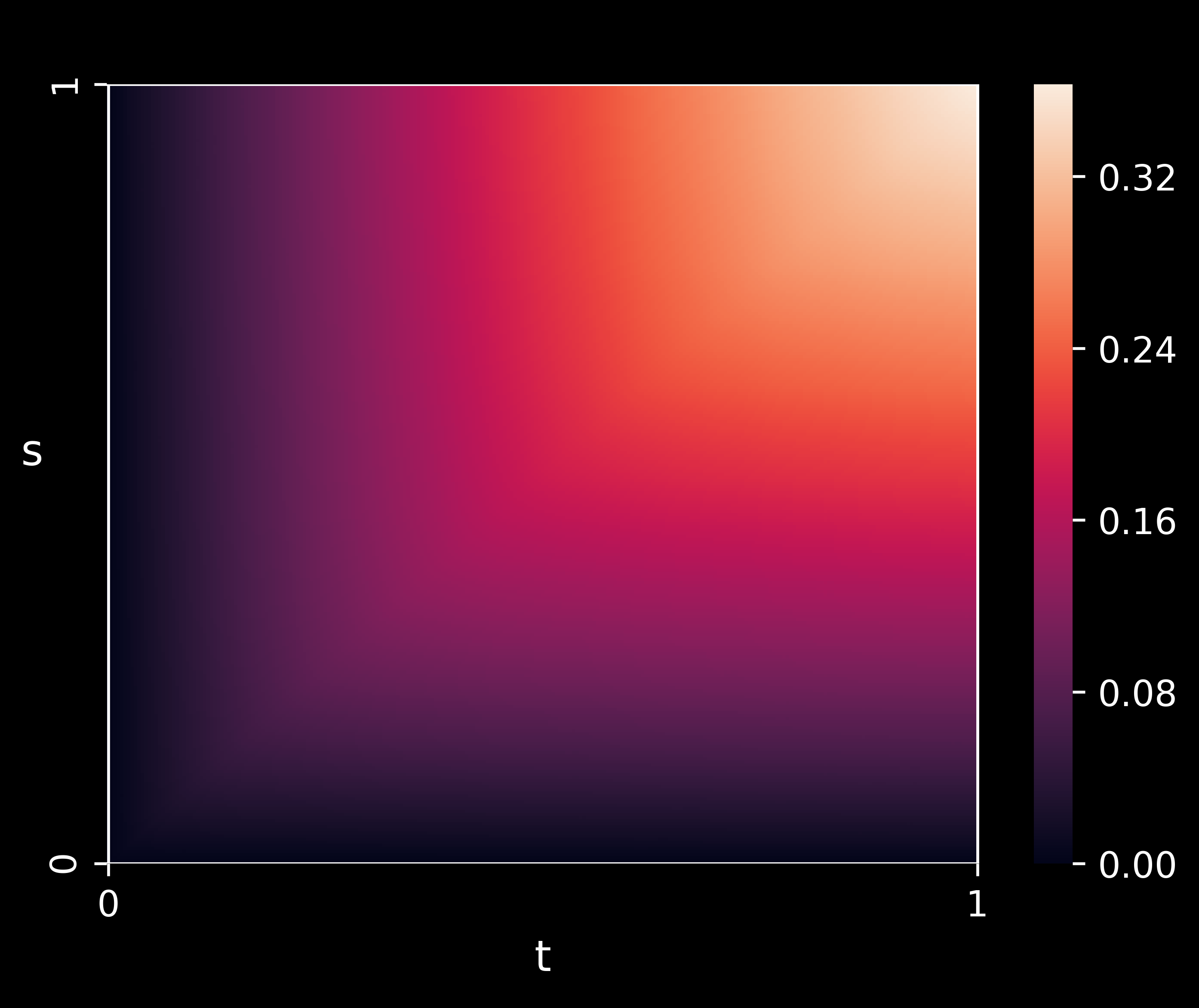}
    \label{fig:MaxK}
\end{subfigure}
\vspace{2mm}

\begin{subfigure}[b]{0.325\textwidth}
    \centering
    \includegraphics[height=1.4in,width=1.85in]{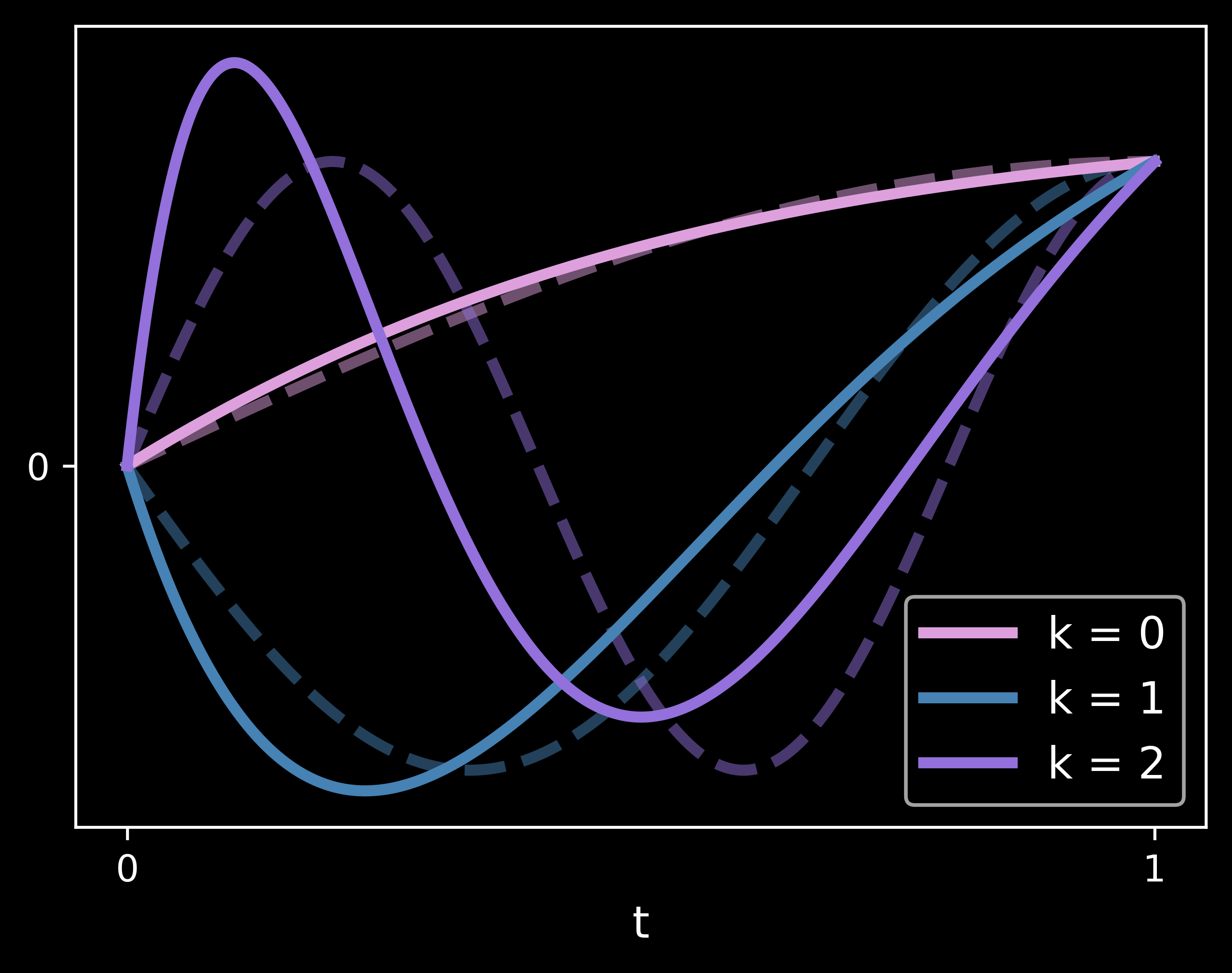}
    \label{fig:Avg}
\end{subfigure}
\begin{subfigure}[b]{0.325\textwidth}
    \centering
    \includegraphics[height=1.4in,width=1.85in]{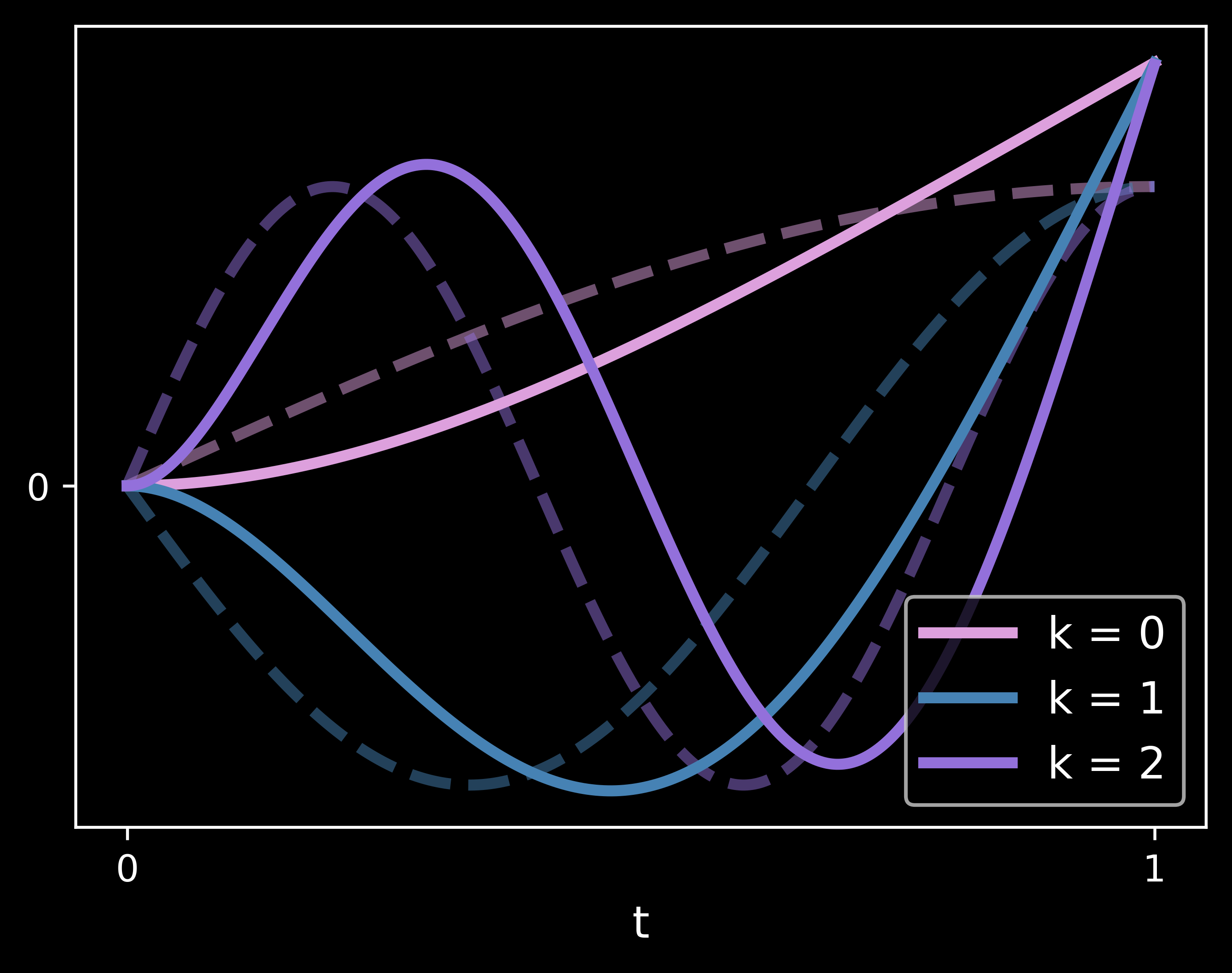}
    \label{fig:Int}
\end{subfigure}
\begin{subfigure}[b]{0.325\textwidth}
    \centering
    \includegraphics[height=1.4in,width=1.85in]{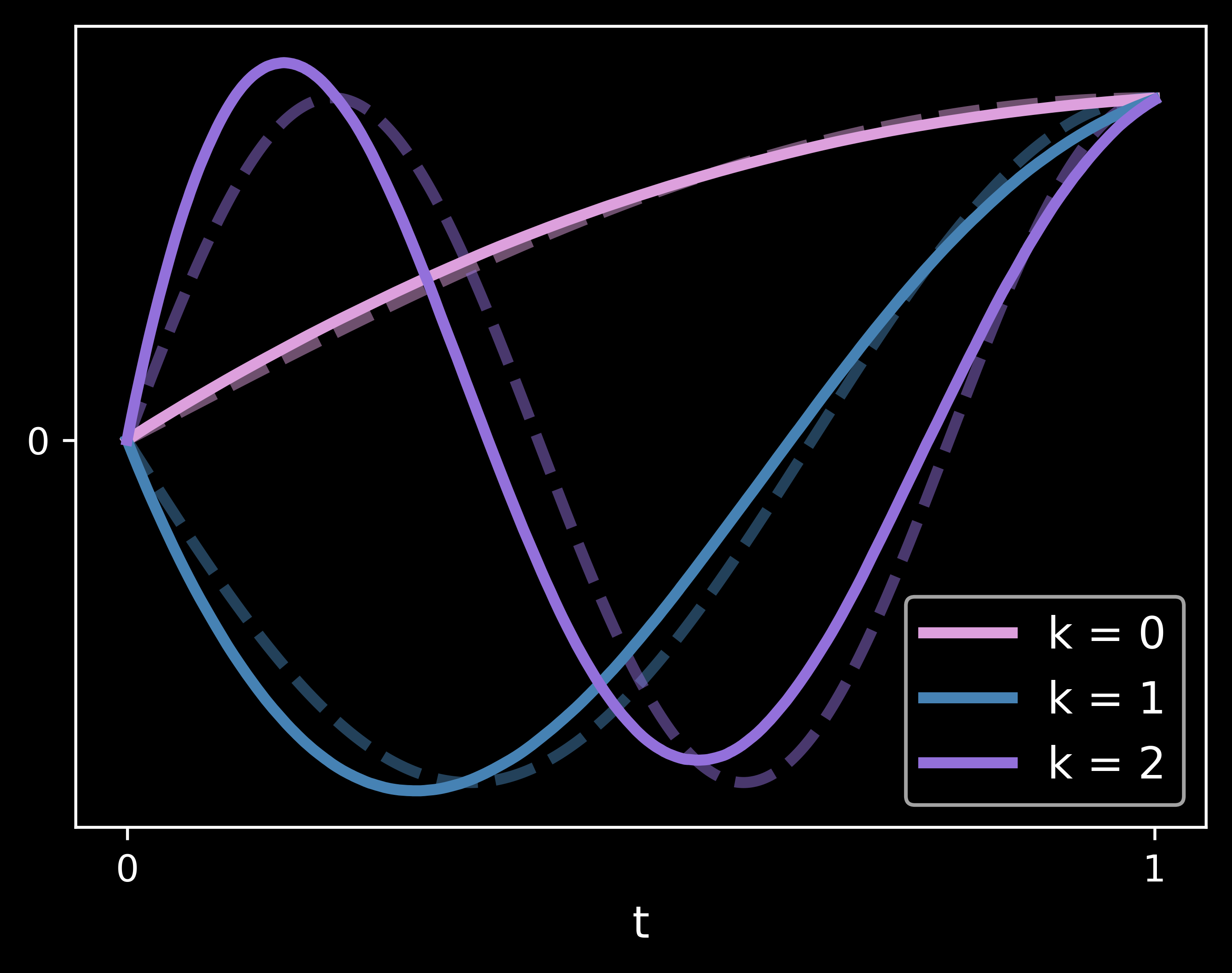}
    \label{fig:Int}
\end{subfigure}
\vspace{1mm}
\begin{center}
    \footnotesize{
    \textit{Solid lines: transformed
     path.  Dashed lines: Brownian motion.}}
\end{center}
\vspace{-3mm}
\end{figure}

\begin{example} Consider the running maximum  functional $y_t= 
f(X_t) = \max_{0 \le s \le t} x_s$.  \Cref{fig:max3D} provides an illustration in the $(t,X,Y)$ plane. The mean function is in this case non-zero and$-$using, e.g., the reflection principle$-$given by $\E^{\Q}[y_t] = \sqrt{\frac{2}{\pi} t}$. The covariance kernel admits an explicit yet complicated expression \cite{Benichou}, 
$$\kappa_Y(s,t) = \frac{s}{2} + \frac{\sqrt{s(t-s)}-2\sqrt{st} + t \arcsin(\sqrt{s/t})}{\pi}, \quad s \le t.$$
\Cref{fig:MaxK} displays the covariance kernel (top) and first eigenfunctions (bottom). The latter turns out to be quite close to the eigenfunctions of Brownian motion. 

\end{example}

\begin{figure}[H]
    \centering
    \caption{Running maximum functional for two trajectories.}
    \vspace{-2mm}
    \includegraphics[scale =0.2]{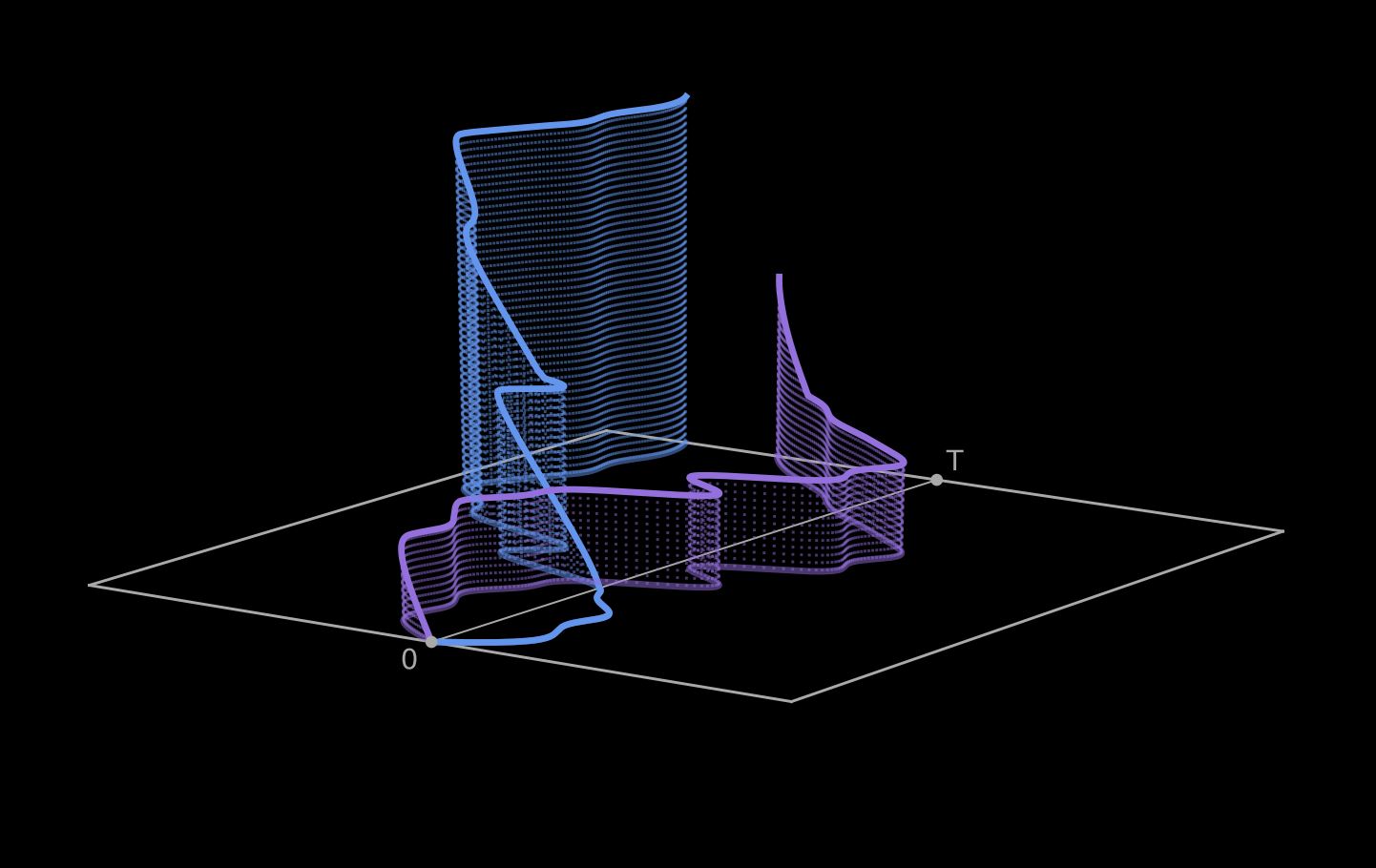}
    \label{fig:max3D}
\end{figure}
\subsection{Numerical Results} \label{ssec: numResult}
Let us compare the $L^2(\Q \otimes \, dt)$ error  
$\lVert Y^{K,\frakF} - Y \rVert^2_{*} $ in the Brownian case for the two avenues discussed at the beginning of this section. When $Y^{K,\frakF} = (f \circ \pi^{K,\frakF})(X) $, the error is calculated using Monte Carlo simulations. As in  \Cref{sec:numResultX}, we choose $T=1$, $N=10^4$ and $K\in \{1,\ldots,128\}$.  
\Cref{fig:L2Error} displays the result for the running maximum, integral and average functionals. We also add the Brownian motion itself, corresponding to the identity functional $f(X)=X$. We observe a clear improvement when projecting the transformed path. Moreover, it comes as no surprise that smooth functionals (integral, average) exhibits a faster rate of convergence than the running maximum, highly sensitive to local behaviours of a path. 

\begin{figure}[H]
    \centering
    \caption{$L^2(\Q \otimes dt)$ approximation errors. }
    \vspace{-2mm}
    \includegraphics[scale =0.38]{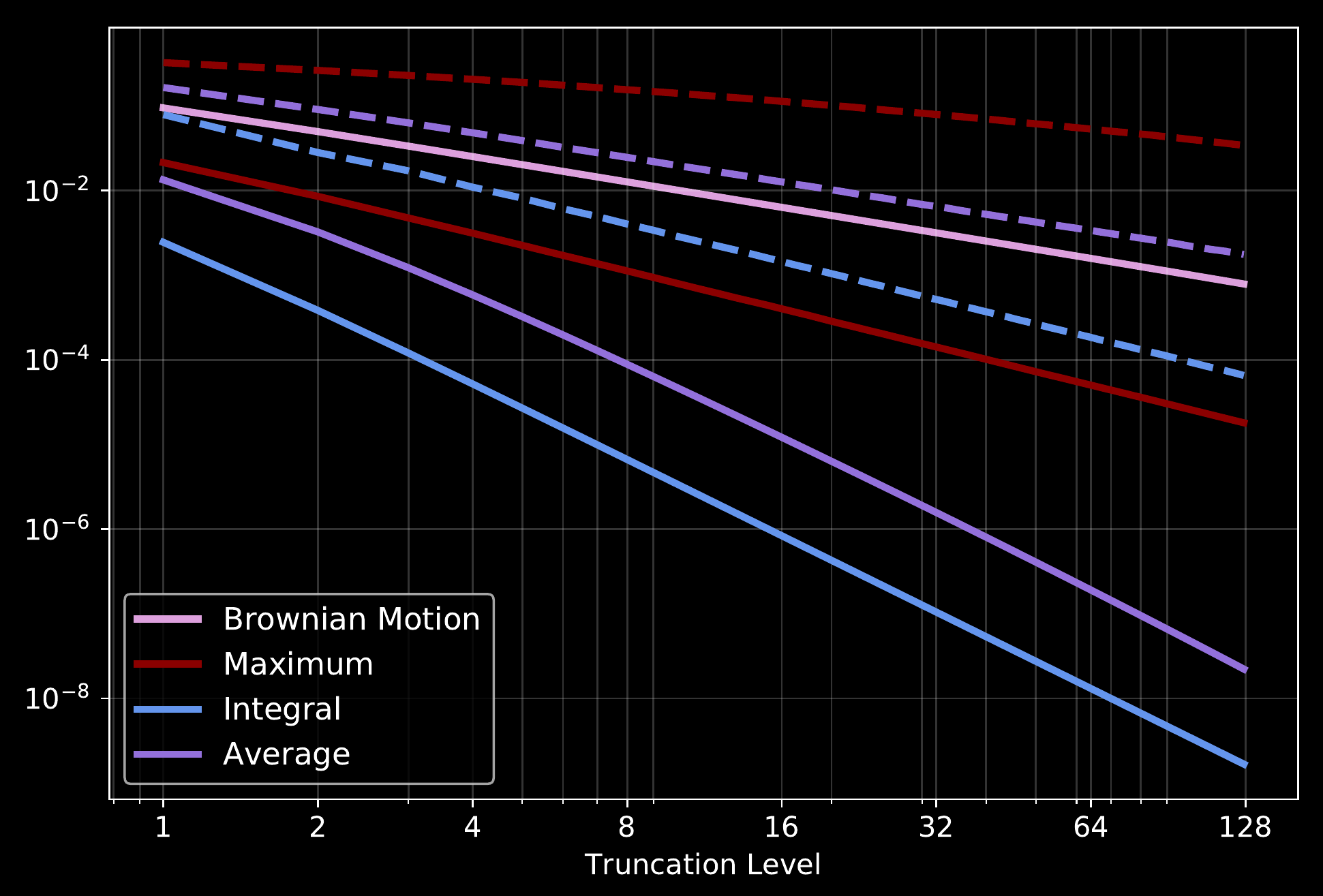}
    \label{fig:L2Error}
    \\
    
    \vspace{1mm}
    
    \footnotesize{
    \textit{ Dashed lines: functionals of projected paths. \\Solid lines: projected functionals.}}
\end{figure}




\subsection{Discussion: Projection of Functionals using the Signature} \label{sec:sigFunc}

Another approximation of $Y=f(X)$ can be obtained by combining  signature functionals in a linear fashion. For instance, one can consider all the words of length less than $\bar{K} \in \N$, giving the approximation
  \begin{equation}\label{eq:sigPayoff}
      y^{K,\calS}_t := \sum_{l(\alpha)\, \le \, \bar{K}} \xi^f_{\alpha} \calS_{\alpha}(X_t), 
  \end{equation}
  \vspace{-3mm}
  
  with $K = |\{\alpha \, | \, l(\alpha)\le \bar{K}\}|=2^{\bar{K}+1}-1.$
  The coefficients $\xi^f_{\alpha}$ may depend on  $X_0$  only and can be  calculated by either regressing  $Y$ against the signature elements or using a Taylor formula for functionals \cite{LittererOberhauser} when $f$ is smooth (in the Dupire sense) and $X$ is a diffusion. 
     In the literature,   $\eqref{eq:sigPayoff}$ is referred to  as  \textit{polynomial functional}  \cite{LittererOberhauser} or \textit{signature payoff} \cite{Szpruch,LyonsNum} in finance. 
     
     The appeal of such projection comes from the fact that polynomial functionals are dense in the space of continuous functionals restricted to  paths of bounded variation; see Theorem 5 in  \cite{LittererOberhauser}.  
     On the other hand, despite the existence of   packages to calculate the signature of discrete time paths (e.g., \texttt{iisignature} and \texttt{esig} in Python), the projection in $\eqref{eq:sigPayoff}$ is still  challenging from a computational perspective. 
     Indeed, contrary to the reconstruction in  \Cref{sec:sigLegendre}, the signature functionals have to be known at \textit{every} intermediate time. Also, as there is a priori no recipe to select words up to a given length,  we must retain all of them so the number of elements doubles every time a layer is added.

\section{Applications} \label{sec:application}
We now illustrate the benefits of the Karhunen-Loève expansion on functionals for the pricing of exotic derivatives. 
We slightly change notations and write $W$ for the coordinate process. The path $X$ now represents the stock price where for simplicity, we employ the Black-Scholes model with zero interest rate. That is,  $\Q$ is the Wiener measure and $x_t = x_0 \calE_t(\sigma  W)$, where $x_0>0$ is fixed and  $\calE$ denotes the stochastic exponential.
Notice, however, that our method applies to any dynamics of the underlying, possibly multidimensional or involving jumps. As we shall see below, what matters is whether the functional  in the payoff, say, $f$, generates square-integrable paths $Y=f(X)$ for the covariance kernel to be well-defined.\footnote{Further work would include a treatment of exotic payoffs depending on \textit{several} functionals, although the latter can be usually combined. 
Take for instance \textit{range options} that entail the difference between the running maximum $\overline{f}(X_t) =\max_{0\le s \le t}x_s$ and minimum $\underline{f}(X_t) =\min_{0\le s \le t}x_s$. Then one simply sets $f = \overline{f} - \underline{f}$. }  

Let $\calM \subseteq \R$, $\calT \subseteq [0,T]$ be a finite set of option parameters and maturities, respectively. 
We seek to approximate the price surface
$p: \calM \times \calT \to \R$,  $p(m,\tau) = \E^{\Q}[ \varphi_m(X_\tau)]$, where the payoffs  $\varphi_m(X_{\tau}) = (h_m \circ  f)(X_{\tau})$ depends on a parameter $m \in \calM$.  For example, a call option on $Y=f(X)$  is obtained with $h^{\text{Call}}_m(y) := (y-m x_0)^{+}$ and $m$ is the \textit{moneyness} of the option.

The standard Monte Carlo  approach (MC) consists of simulating the underlying path  on a partition  $\Pi_{N} = \{0=t_0 < t_1 < ... < t_N=T \}$  that contains $\calT$ and compute the price as $p^{N,J}(m,\tau) = \frac{1}{J}\sum_{j=1}^J \varphi_m(X^{N,j}_\tau)$, $J\in \N$. 
In contrast, the \textit{Karhunen-Loève Monte Carlo method} (KLMC) samples $Y=f(X)$ directly and computes the price surface  using  the representation $p(m,\tau) = \E^{\Q}[ h_m(y_\tau)]$.  
We now describe the method in more depth. 
\subsection{The KLMC Algorithm}
We assume for simplicity that $Y$ has zero mean, otherwise minor changes  must be made for the KL expansion; see \cref{rem:center}. 
First, we simulate trajectories
$Y^j=(f(X_{t_n}^j))_{t_n \in \Pi_{N_{\text{off}}}}$, $j=1,...,J_{\text{off}}$ with  $J_{\text{off}}, N_{\text{off}} \in \N$,  and compute the eigenfunctions of $\kappa^{N_{\text{off}}}_Y$ as in  $\eqref{eq:eigendecomp}$. Next, for $k=1,...,K$,   we estimate the sample quantile function $\Phi_k^{-1}:[0,1]\to \R$ of $\xi_k = (Y,F_k)$ by employing method N\textsuperscript{o} $\! 7$ in \cite{HyndmanFan}.  
The coefficients $(\xi_k )$ can thereafter be simulated using inverse transform sampling \cite{Devroye}. 
  Notice that these steps can be done in an offline phase so $(\Phi_k^{-1})$ can be reused for other options contingent upon the functional $f$; see \Cref{sec:numResultX}. 
  
In the online phase, we simulate $J \in \N$ transformed paths using inverse transform sampling for $\xi_k$.   
Finally, the price surface is calculated using Monte Carlo.  The procedure is summarized in \Cref{alg:klmc}. 
It should be noted that although the coefficients $(\xi_k)$ are orthogonal in $L^2(\Q)$, they may well be \textit{dependent} when $Y$ is non-Gaussian. While the marginals of $(\xi_k)_{k\le K}$ are fitted properly, 
the dependence is omitted as generating dependent random vectors with unknown joint distribution is highly non-trivial. Nevertheless, this simplification doesn't induce a bias in the obtained prices as we shall see in the numerical experiments.  

\begin{algorithm}[H]
\caption{(KLMC) }\label{alg:klmc}
\begin{itemize}
\vspace{-2mm}
\item \textbf{Offline}: Given $f$, $K$,  $J_{\text{off}}$, $N_{\text{off}}$ 
\begin{enumerate}
\setlength \itemsep{0.2ex}
\vspace{-2mm}
\item Simulate trajectories $Y^j=(f(X_{t_n}^j))_{t_n \in \Pi_{N_{\text{off}}}}$, $j=1,...,J_{\text{off}}$
\item Compute $\kappa^{N_{\text{off}}}_Y$ (closed-form or from the sample  $(Y^j)$)
\item Solve the eigenvalue problem $\eqref{eq:eigendecomp}$ to obtain  $(\lambda^{\frakF}_k,F_k)$
\item Using $(Y^j)$, estimate the quantile functions  $\Phi_k^{-1}$, $k\le K$ 
\end{enumerate}
\vspace{-2mm}

\item \textbf{Online}: Given $J, \ \calM, \ \calT$ 
\begin{enumerate}
\setlength \itemsep{0.2ex}
\vspace{-2mm}
\item Simulate $\xi^j_k = \Phi^{-1}_k(u_k^j)$, $(u_k^j) \overset{i.i.d.}{\sim} U(0,1)$, $j \le J$, $k\le K$
\item Compute $y_\tau^{K,\frakF,j} =\sum_{k=1}^K  \xi^j_k    F_k(\tau)$, $\tau \in \calT$, $j\le J$ 
\item Estimate the price surface $p^{K,\frakF,J}(m,\tau) := \frac{1}{J}\sum_{j=1}^{J} h_m(y_\tau^{K,\frakF,j}).$
\end{enumerate}
\end{itemize}
\vspace{-3mm}
\end{algorithm}

\subsection{Numerical Results} \label{sec:numResultPrice}
First, we build the price surface for  Asian and lookback call options, i.e. by choosing $h_m = h^{\text{Call}}_m$ and the running maximum and time average as underlying functional, respectively. Of course, the put option price surface can  be retrieved thanks to  put-call parity. We also consider Up \& Out digital options, that is $f(X_t)=\max_{0\le s \le t}x_s$  and   $h^{\text{UO}}_m(y):=\mathds{1}_{\{y \ \le \ m x_0\}}$. The parameter $m\ge 1$ thus represents the barrier of the option relative to the spot price. 
We can therefore reuse the quantile functions computed for the lookback call options.

The parameters are $(x_0,\sigma,N_{\text{off}},J_{\text{off}},J) = (100, 0.2, 10^3, 2^{17}, 2^{19})$, 
 $T=1$ year and $\calT =  \{\frac{1}{52},\frac{2}{52},\ldots,1\}$ (weekly maturities). 
The moneyness and barrier levels are respectively $\calM^{\text{Call}} = \{0.75,0.80,\ldots,1.25\}$ and $\calM^{\text{UO}} = \{1.05,1.10,\ldots,1.50\}$. 
We assess  accuracy  in the mean square sense, namely by computing 
$$
\text{MSE} =  \frac{1}{|\calM|\ |\calT|}\sum_{(m,\tau) \in \calM \times \calT} |p^{\text{(B)}}(m,\tau) - \hat{p}(m,\tau)|^2. 
$$ The function $\hat{p}$ is the approximated price and $p^{\text{(B)}}$ a benchmark obtained using a standard Monte Carlo with $40 \cdot |\calT| = 2080$ time steps and same number of simulations. 
\interfootnotelinepenalty=10000 
  \cref{tab:results} displays the MSE, runtime (online phase) and number of variates per simulated path ($K$ and $N$ for the KLMC and MC method, respectively).\footnote{The experiments have been made on a personal computer; see \href{https://github.com/valentintissot/KLMC.git}{https://github.com/valentintissot/KLMC} for an implementation. The offline phase takes about 10 seconds per  functional.}  Notice that we increase $K,N$  for the lookback call and Up \& Out digital option as the running maximum has a slower rate of $L^2(\Q\otimes dt)$ convergence  as seen in \Cref{fig:L2Error} for the Brownian case.  The KLMC method constantly yields a lower MSE and runtime.  For the Asian call option, note that the number of variates per path is less than the number of maturity points and KLMC method, which couldn't be done with the MC method. 
  
  \vspace{-1mm}

\begin{table}[H]
    \centering
\caption{Mean squared errors and runtime (seconds)}
\vspace{-3mm}
\begin{tabular}{ccccccc} 
\hline  & \multicolumn{3}{c}{KLMC}  & \multicolumn{3}{c}{MC} \\  \hline
Option & K & MSE  & Time & N & MSE & Time \\
\hline  
 Asian Call & $40$ & 1.50e-04 & 2.26 & $ 52$ & 3.10e-04 & 2.40\\ 
Lookback Call   & $100$ & 1.15e-02 & 4.47& $4\cdot 52$ & 1.92e-01 & 7.05 \\ 
Up \& Out Digital   & $100$ & 1.80e-04 & 4.40 & $4\cdot 52$ & 1.90e-04 & 6.88\\ \hline
\end{tabular}
\label{tab:results}
\end{table}

\section*{Conclusion}
This paper sheds further light on the approximation of path functionals. After a thorough review of Hilbert projections and a connection with the path signature, we show the power of the Karhunen-Loève expansion to parsimoniously simulate path-dependent payoffs. 
Further work would include the use of copulas in the  KLMC algorithm to capture the dependence between the $L^2([0,T])$ coefficients of the Karhunen-Loève expansion and a performance comparison with signature-based methods. 

\bibliography{main.bib}

\begin{thebibliography}{10}

\bibitem{Acworth}
P.~A. Acworth, M.~Broadie, and P.~Glasserman.
\newblock {\em A Comparison of Some Monte Carlo and Quasi Monte Carlo
  Techniques for Option Pricing}.
\newblock Springer New York, New York, NY, 1998.

\bibitem{Szpruch}
I.~P. Arribas, C.~Salvi, and L.~Szpruch.
\newblock \textnormal{Sig-SDEs model for quantitative finance}, 2020.

\bibitem{Benichou}
O.~B\'enichou, P.~L. Krapivsky, C.~Mej\'{\i}a-Monasterio, and G.~Oshanin.
\newblock Temporal correlations of the running maximum of a brownian
  trajectory.
\newblock {\em Phys. Rev. Lett.}, 117:080601, Aug 2016.

\bibitem{Brown}
B.~Brown, M.~Griebel, F.~Y. Kuo, and I.~H. Sloan.
\newblock On the expected uniform error of geometric brownian motion
  approximated by the l\'evy-ciesielski construction, 2017.

\bibitem{Hull}
J.~Cao, J.~Chen, J.~C. Hull, and Z.~Poulos.
\newblock Deep learning for exotic option valuation.
\newblock {\em SSRN}, 2021.

\bibitem{Carr}
P.~Carr and D.~Madan.
\newblock Option valuation using the fast fourier transform.
\newblock {\em Journal of Computational Finance}, 2:61--73, 1999.

\bibitem{Devroye}
L.~Devroye.
\newblock {\em Non-Uniform Random Variate Generation}.
\newblock Springer-Verlag, 1986.

\bibitem{Foster}
J.~Foster, T.~Lyons, and H.~Oberhauser.
\newblock An optimal polynomial approximation of brownian motion.
\newblock {\em SIAM J. Numer. Anal.}, 58:1393--1421, 2020.

\bibitem{Geng}
X.~Geng.
\newblock Reconstruction for the signature of a rough path.
\newblock {\em Proceedings of the London Mathematical Society},
  114(3):495--526, 2017.

\bibitem{Ghanem}
R.~Ghanem and P.~Spanos.
\newblock {\em Stochastic Finite Elements: A Spectral Approach}.
\newblock Springer, New York, NY, 1991.

\bibitem{Hambly}
B.~Hambly and T.~Lyons.
\newblock Uniqueness for the signature of a path of bounded variation and the
  reduced path group.
\newblock {\em Annals of Mathematics}, 171(1):109–167, Mar 2010.

\bibitem{HyndmanFan}
R.~J. Hyndman and Y.~Fan.
\newblock Sample quantiles in statistical packages.
\newblock {\em The American Statistician}, 50(4):361--365, 1996.

\bibitem{Karhunen}
K.~Karhunen.
\newblock \textnormal{Über lineare Methoden in der
  Wahrscheinlichkeitsrechnung}.
\newblock {\em Annales Academiae scientiarum Fennicae}, 37, 1947.

\bibitem{LittererOberhauser}
C.~Litterer and H.~Oberhauser.
\newblock On a {Chen-Fliess} approximation for diffusion functionals.
\newblock {\em Monatshefte fur Mathematik}, 175(4):577--593, 2014.

\bibitem{Loeve}
M.~Lo\`eve.
\newblock Fonctions al\'eatoires du second ordre.
\newblock {\em Gauthier Villars}, 1948.
\newblock Supplement to "Processus Stochastique et Mouvement Brownien" from
  Paul L\'evy.

\bibitem{Lyons}
T.~J. Lyons, M.~Caruana, and T.~Lévy.
\newblock {\em Differential Equations Driven by Rough Paths}.
\newblock Berlin: Springer, 2007.

\bibitem{LyonsNum}
T.~J. Lyons, S.~Nejad, and I.~P. Arribas.
\newblock Numerical method for model-free pricing of exotic derivatives using
  rough path signatures.
\newblock {\em Applied Mathematical Finance}, 26(6):583--597, 2019.

\bibitem{Mercer}
J.~Mercer.
\newblock Functions of positive and negative type, and their connection with
  the theory of integral equations.
\newblock {\em Philosophical Transactions of the Royal Society of London.},
  209:415--446, 1909.

\bibitem{Pages}
G.~Pagès and J.~Printems.
\newblock Functional quantization for numerics with an application to option
  pricing.
\newblock {\em Monte Carlo Methods and Applications}, 11(4):407--446, 2005.

\bibitem{reinhardt}
H.~J. Reinhardt.
\newblock {\em Analysis of approximation methods for differential and integral
  equations}.
\newblock Applied mathematical sciences ; v. 57. Springer-Verlag, New York,
  1985.

\bibitem{Schwartz}
E.~S. Schwartz.
\newblock The valuation of warrants: Implementing a new approach.
\newblock {\em Journal of Financial Economics}, 4(1):79--93, 1977.

\bibitem{Solin}
A.~Solin and S.~Särkkä.
\newblock Hilbert space methods for reduced-rank gaussian process regression.
\newblock {\em Statistics and Computing}, 30, 03 2020.

\bibitem{Szego}
G.~Szegö.
\newblock Orthogonal polynomials.
\newblock {\em American Mathematical Society}, XXII, 1975.
\newblock Fourth edition.

\end{thebibliography}
\end{document}